\documentclass[11pt,aps,pra,onecolumn,superscriptaddress,floatfix,
nofootinbib,showpacs,longbibliography]{revtex4-2}
\usepackage[utf8]{inputenc}  
\usepackage[T1]{fontenc}     
\usepackage[british]{babel}  
\usepackage[sc,osf]{mathpazo}
\usepackage{libertineRoman}  
\usepackage[colorlinks=true, citecolor=blue, urlcolor=blue]{hyperref}  
\usepackage{graphicx}
\usepackage{diagbox}
 \usepackage{enumitem}
 \usepackage[babel]{microtype}  
\usepackage{amsmath,amssymb,amsthm,bm,amsfonts,mathrsfs,bbm} 

\usepackage{xspace}  
\usepackage{pgf,tikz}
\usepackage{xcolor}
\usepackage{multirow}
\usepackage{array}
\usepackage{bigstrut}
\usepackage{braket}
\usepackage{color}
\usepackage{natbib}
\usepackage{multirow}
\usepackage{mathtools}
\usepackage{float}
\usepackage{xcolor,colortbl}
\usepackage{physics}
\usepackage{amsmath}
\usepackage{color}
\usepackage[justification=justified, format=plain]{subcaption}
\usepackage[justification=raggedright]{caption}

\newcommand{\be}{\begin{equation}}
\newcommand{\ee}{\end{equation}}
\newcommand{\ba}{\begin{eqnarray}}
\newcommand{\ea}{\end{eqnarray}}

\newtheorem{theorem}{Theorem}
\newtheorem{corollary}{Corollary}
\newtheorem{definition}{Definition}
\newtheorem{proposition}{Proposition}

\newtheorem{example}{Example}
\newtheorem{remark}{Remark}
\newtheorem{lemma}{Lemma}

\def\>{\rangle}
\def\<{\langle}

\begin{document}
	
\title{Absolute Schmidt number: characterization, detection and resource-theoretic quantification}	

\author{Bivas Mallick }
\email{bivasqic@gmail.com}
\affiliation{S. N. Bose National Centre for Basic Sciences, Block JD, Sector III, Salt Lake, Kolkata 700 106, India}

\author{Saheli Mukherjee}
\email{mukherjeesaheli95@gmail.com}
\affiliation{S. N. Bose National Centre for Basic Sciences, Block JD, Sector III, Salt Lake, Kolkata 700 106, India}

\author{Nirman Ganguly}
\email{nirmanganguly@hyderabad.bits-pilani.ac.in,nirmanganguly@gmail.com}
\affiliation{Department of Mathematics, Birla Institute of Technology and Science, Pilani, Hyderabad Campus, Jawahar Nagar, Kapra Mandal, Medchal District, Telangana 500078, India}

\author{A. S. Majumdar}
\email{archan@bose.res.in}
\affiliation{S. N. Bose National Centre for Basic Sciences, Block JD, Sector III, Salt Lake, Kolkata 700 106, India}
	
\begin{abstract}
The dimensionality of entanglement, quantified by the Schmidt number, is a valuable resource for a wide range of quantum information processing tasks. In this work, we introduce the notion of the absolute Schmidt number, referring to states whose Schmidt number cannot be increased by any global unitary transformation. We  provide a  characterization of the set of arbitrary-dimensional states whose Schmidt number is invariant under all global unitaries. Our approach enables us to develop both witness-based and moment-based techniques to detect nonabsolute Schmidt number states which
could provide significant operational advantages through Schmidt number 
enhancement by global unitaries. We next formulate two resource-theoretic 
measures of nonabsolute Schmidt number states, based respectively on Schmidt number witness and robustness, and demonstrate an operational utility of
the latter in a channel discrimination task. Finally, we extend our analysis to quantum channels by introducing a new class of channels that possess the absolute Schmidt number property. We  derive a necessary and sufficient condition for identifying when a channel has the absolute Schmidt number property,   confining our analysis to the class of covariant channels.
\end{abstract}
\maketitle

\section{Introduction}

The rapid advancement of quantum science and technology has propelled quantum correlations to the frontiers of modern quantum information research. 
Within the ambit of quantum correlations, entanglement~\cite{horodecki2009quantum} which was first emphasized through the Einstein–Podolsky–Rosen paradox \cite{einstein1935can}, serves as a key resource that enables a variety of information processing advantages such as quantum communication ~\cite{PhysRevLett.69.2881,PhysRevLett.70.1895,PhysRevLett.126.250501,piveteau2022entanglement,PhysRevLett.134.020802}, quantum cryptography ~\cite{PhysRevLett.67.661,lo1999unconditional,yin2020entanglement}, quantum computational speedups~\cite{jozsa2003role,PhysRevA.101.012349}, and randomness generation~\cite{pironio2010random,PhysRevLett.120.010503}. These far-reaching applications have motivated the development of a comprehensive resource-theoretic understanding of entanglement, wherein separable states are free states, entangled states serve as resources, and local operations along with classical communication (LOCC) define the set of free operations~\cite{chitambar2019quantum}.

In the standard setting, separable states provide no advantage since they cannot be used to generate entanglement under LOCC. This in turn raises a fundamental question: if one is allowed to employ global unitary operations, can a separable state be transformed into an entangled one? In general, the answer is positive for pure product states \cite{thirring2011entanglement,guan2014entangling}. Moreover, there exists certain mixed separable states that can indeed be converted into entangled states via a suitably chosen global unitary operation, making them potentially useful in a wide range of quantum information processing tasks  \cite{ishizaka2000maximally,sackett2000experimental,rauschenbeutel2000step,kraus2001optimal,kastoryano2011dissipative}. However, this is not universally the case, \textit{i.e.}, there exists a class of mixed separable states that remain separable under \emph{all} global unitary operations. These states, known as \textit{absolutely separable} states~\cite{PhysRevA.88.062330, PhysRevA.63.032307, ganguly2014witness,  serrano2024absolute,louvet2025nonequivalence,arunachalam2014absolute,PhysRevA.58.883,abellanet2025sufficient,patra2021efficient,PhysRevA.103.052431}, are characterized by their intrinsic inability to generate entanglement, even under the most general reversible transformations.

Beyond the primary presence of entanglement, the \emph{structure} and \emph{dimensionality} of entanglement play crucial roles in determining how effectively a quantum state performs in information processing tasks \cite{terhal2000schmidt,liu2023characterizing,liu2024bounding,krebs2024high,zhang2024analyzing,mukherjee2025measurement,mukherjee2025certifying}. The number of degrees of freedom participating in entanglement is quantified by the Schmidt number of a bipartite state~\cite{terhal2000schmidt}. States with higher Schmidt number have been shown to provide enhanced advantages in tasks such as quantum communication~\cite{zhang2025quantum}, channel discrimination~\cite{bae2019moreclar}, quantum control~\cite{kues2017chip}, and quantum key distribution~\cite{cerf2002security}. Therefore, from a resource-theoretic viewpoint, there are situations in which $d$-dimensional quantum states with Schmidt number at most $r$ (where $r<d$) are regarded as free, while those with Schmidt number greater than $r$ are treated as resourceful states.

Considering the operational benefits offered by higher Schmidt number states, an important yet unexplored question arises in this context: \emph{Can the Schmidt number of a state be enhanced by global unitary operations?} While the Schmidt number is invariant under local unitaries \cite{terhal2000schmidt}, global unitaries can, in general, modify the entanglement dimensionality of a bipartite state. This leads to a fundamental distinction between states whose Schmidt number can be increased by global unitaries and those for which no such enhancement is possible. Motivated by this, in this work, we introduce the notion of the \textit{absolute Schmidt number} for quantum states. We posit that a quantum state belongs to $r$\text{-}$\mathbb{ABSN}$ if its Schmidt number cannot be increased by any global unitary operation. This notion generalizes the concept of absolute separability to the full hierarchy of entanglement dimensionality. We then provide a detailed characterization of the set $r$\text{-}$\mathbb{ABSN}$ by proving that it forms a convex and compact set. Consequently, by the Hahn--Banach theorem \cite{holmes2012geometric}, one can construct a witness capable of detecting states that do not belong to $r$\text{-}$\mathbb{ABSN}$. 

While witness-based techniques can, in principle, be applied to states in arbitrary dimensions, they typically rely on partial prior information about the state to construct an appropriate witness \cite{matsunaga2025detecting,guhne2009entanglement,guhne2006nonlinear}. To overcome this limitation, we further seek to identify signatures of states that do not belong to $r$\text{-}$\mathbb{ABSN}$ without imposing any structural assumptions on the state. For this purpose, we introduce a moment-based framework \cite{elben2020mixed,yu2021optimal, aggarwal2024entanglement,  mallick2024assessing,mallick2025efficient,mukherjee2025detecting,chakrabarty2025probing,mallick2025higher,mallick2026detection} that circumvents the need for full state tomography \cite{huang2020predicting,aaronson2018shadow,aaronson2019gentle}. We further analyze our detection scheme with several illustrative examples. 

Apart from identifying states that do not belong to $r$\text{-}$\mathbb{ABSN}$, the convexity and compactness of this set also make it suitable for quantitative analysis. In this context, we introduce witness-based measure and robustness measure for quantifying non-$r$-absolute Schmidt-number property of a state. We show that these witness-based and robustness measures satisfy several desirable properties of a valid measure, including positivity, invariance under local unitaries, monotonicity under free operations, and convexity. We further demonstrate an operational utility
of robustness in terms of the quantitative advantage offered in a channel
discrimination task. It is worth mentioning that the measures introduced in this work provide a broader extension than the one proposed in \cite{patra2023resource} for quantifying non-absolute separability, as our formulation captures the entire hierarchy of entanglement dimensionality.
It may also be noted that during the completion of the present work, a complementary approach for analyzing absolute Schmidt numbers has been proposed \cite{Jofre_ASN_2026}.

Next, we introduce and investigate a new class of quantum channels that exhibit an absolute Schmidt number property. In realistic scenarios, environmental noise may be sufficiently strong to degrade the entanglement dimensionality of any quantum state passing through the channel, i.e., regardless of the input state, the output produced by the channel always belongs to the set $r$\text{-}$\mathbb{ABSN}$. We refer to such channels as \textit{$r$-absolute Schmidt number channels}. After introducing this class, we derive a necessary and sufficient condition for identifying such channels that are confined to the class of covariant channels. We then provide an illustrative example in support of our detection scheme.

The structure of the paper is as follows. In Section \ref{s2}, we present a concise overview of the Schmidt number of a density matrix. In Section \ref{s3}, we introduce the notion of $r$-absolute Schmidt number states and provide a detailed characterization of the set $r$\text{-}$\mathbb{ABSN}$, including criteria for detecting states that do not belong to $r$\text{-}$\mathbb{ABSN}$ using both witness-based and moment-based approaches. This section also examines several properties and illustrative examples of states that do belong to $r$\text{-}$\mathbb{ABSN}$. Section \ref{s4} focuses on quantifying the non-absolute Schmidt number property of a state by introducing two well-established measures: a witness-based measure and a robustness measure. We then introduce the notion of absolute Schmidt number channels and derive a necessary and sufficient condition characterizing such channels in Section \ref{s5}. 
 Finally, Section \ref{s7} summarizes the main results of the paper and outlines promising directions for future research.

\section{Schmidt number of a density matrix}\label{s2}
Consider a bipartite quantum system comprising two subsystems, $A$ and $B$, associated with the joint Hilbert space $\mathcal{H}_A \otimes \mathcal{H}_B$. Let, ${\mathcal{H}}_A = {\mathcal{H}}_B = {\mathbf{C}}^d$, then any bipartite pure state $\ket{\psi}_{AB} \in {\mathbf{C}}^d \otimes {\mathbf{C}}^d$ of this composite system admits a Schmidt decomposition of the form \cite{nielsen2010quantum}:
\begin{equation}
\ket{\psi}_{AB} = \sum_{i=1}^{k} \sqrt{\lambda_i} \ket{i}_A \ket{i}_B,
\end{equation}
where \( \lambda_i \geq 0 \) are the Schmidt coefficients, satisfying the normalization condition \( \sum_{i} \lambda_i = 1 \), and the sets \( \{\ket{i}_A\} \) and \( \{\ket{i}_B\} \) form orthonormal bases of their respective Hilbert spaces. The number of nonzero Schmidt coefficients, denoted by $k$, is referred to as the Schmidt rank (SR) of the state. A bipartite pure state $\ket{\psi}_{AB}$ is classified as a product state if its Schmidt rank is 1; otherwise, it is said to be entangled.

Terhal and Horodecki introduced the concept of Schmidt number (SN) in \cite{terhal2000schmidt}, extending the notion of Schmidt rank (SR) from pure states to mixed states. Let $\mathcal{D}(\mathbb{C}^{d} \otimes \mathbb{C}^{d})$ denote the set of all density operators acting on the bipartite Hilbert space. A mixed bipartite quantum state \( \sigma_{AB} \in \mathcal{D}(\mathbb{C}^{d} \otimes \mathbb{C}^{d}) \) can be represented as a convex combination of pure states:
\begin{equation} \label{decomposition}
    \rho_{AB} = \sum_k p_k \ket{\psi_k}_{AB} \bra{\psi_k},
\end{equation}
where \( \ket{\psi_k}_{AB} \in \mathbb{C}^{d} \otimes \mathbb{C}^{d} \), \( p_k \geq 0 \), and \( \sum_k p_k = 1 \). The \textit{Schmidt number} (SN) of \( \rho_{AB} \) is defined as the minimum of the largest Schmidt rank among all possible pure-state decompositions of the state:
\begin{equation}
    \text{SN}(\rho_{AB}) = \min_{\{p_k, \ket{\psi_k}_{AB}\}} \left[ \max_k \, \text{SR}(\ket{\psi_k}_{AB}) \right]
\end{equation}
where the minimization runs over all ensembles \( \{p_k, \ket{\psi_k}_{AB} \} \) satisfying Eq.\eqref{decomposition} , and \( \text{SR}(\ket{\psi_k}_{AB}) \) denotes the Schmidt rank of the pure state \( \ket{\psi_k}_{AB} \). Consequently, for any bipartite state $\rho_{AB}$, the Schmidt number satisfies the bound $1 \leq \text{SN}(\rho_{AB}) \leq  d$. Moreover, for all separable states \( \rho_{\text{sep}} \), the Schmidt number attains its minimal value, i.e., \( \text{SN}(\rho_{\text{sep}}) = 1 \). Let  $\mathcal{S}_r$ denote the set of all states in $\mathcal{D} ({\mathbf{C}}^d \otimes {\mathbf{C}}^d)$ whose Schmidt number is at most $r$. The set $\mathcal{S}_r$ forms a convex and compact subset within the space of density matrices $\mathcal{D} ({\mathbf{C}}^d \otimes {\mathbf{C}}^d)$ \cite{terhal2000schmidt}. Furthermore, the sets satisfy the nested relation $\mathcal{S}_1  \subset \mathcal{S}_2 ...\subset \mathcal{S}_r$, where $\mathcal{S}_1$ 
  represents the set of separable states.

The set $\mathcal{S}_1$ is completely characterized by positive but not completely positive maps. However for $r>1$, describing the structure of $\mathcal{S}_r$ requires linear maps that are $r$-positive but not $(r+1)$-positive. In particular, a map $\Lambda_{\mathcal{R}}:\mathcal{M}_d \rightarrow \mathcal{M}_d $ must satisfy  
\begin{equation}
    (id_A \otimes \Lambda_{\mathcal{R}})(\rho) \ge 0 
    \qquad \text{for all } \rho \in \mathcal{S}_r ,
\end{equation}
and  
\begin{equation}
    (id_A \otimes \Lambda_{\mathcal{R}})(\sigma) \ngeq 0 
    \qquad \text{for at least one } \sigma \notin \mathcal{S}_r .
\end{equation}
where, $\mathcal{M}_d$ denotes the set of all $d \times d$ complex matrices. An example of such a map is the \emph{Reduction map}~\cite{terhal2000schmidt}, defined as  
\begin{equation} \label{reduction}
    \Lambda_r(\rho) := \operatorname{Tr}(\rho)\, I_d - k\, \rho ,
    \qquad \rho \in \mathcal{M}_d .
\end{equation}
For $r<d$, the map $\Lambda_r$ is $r$-positive but not $(r+1)$-positive whenever the parameter $k$ satisfies  
\begin{equation} \label{rpositive}
    \frac{1}{r+1} < k \le \frac{1}{r} .
\end{equation}
With these necessary preliminaries about the Schmidt number of bipartite quantum states, in the next section, we now proceed to the characterization of the class of absolutely Schmidt number states.

\section{Absolute Schmidt number quantum states} \label{s3}
In this section, we begin by introducing the notion of the absolute Schmidt number for quantum states. We then characterize the set $r$\text{-}$\mathbb{ABSN}$, i.e., the collection of all states whose Schmidt number does not exceed $r$ under the action of any global unitary. Subsequently, we present methods for identifying states that do not have absolute Schmidt number property. We then discuss properties and examples of states which have the absolute Schmidt number property.

\subsection{Characterization and properties of absolute Schmidt number}
\textbf{Definition 1:} We define the set $r$\text{-}$\mathbb{ABSN}$ as 
\begin{equation}
    r\text{-}\mathbb{ABSN} = \{\rho \in \mathcal{D} ({\mathbf{C}}^d \otimes {\mathbf{C}}^d): \text{SN}(\mathcal{U} \rho {\mathcal{U}}^{\dagger} ) \le r ,\hspace{0.1cm}\forall \hspace{0.1cm}\mathcal{U} \in \mathbb{U} (d^2)\}
\end{equation}
where, $\mathcal{U}$ is a global unitary operator i.e., a state $\rho \in \mathcal{D} ({\mathbf{C}}^d \otimes {\mathbf{C}}^d)$ belongs to $r$\text{-}$\mathbb{ABSN}$ iff $\text{SN}(\mathcal{U} \rho {\mathcal{U}}^{\dagger} ) \le r$  for all global unitary operators $\mathcal{U}$ and $\mathbb{U}(d^2)$ denotes the set of all unitary operators acting on the composite Hilbert space $\mathbb{C}^{d}\otimes \mathbb{C}^{d}$ {(see also \cite{Jofre_ASN_2026})}.\\
Based on this definition, we are now in a position to characterize the set $r\text{-}\mathbb{ABSN}$.
\begin{theorem} \label{theorem1}
    The set $r\text{-}\mathbb{ABSN}$ is convex.
\end{theorem} 
\proof  To prove that the set $r\text{-}\mathbb{ABSN}$ is convex, we begin by assuming that both $\rho_{1}$ and $\rho_{2}$ are elements of $r\text{-}\mathbb{ABSN}$. Our goal is to show that for any $p \in [0,1]$, the convex combination $\rho= p \rho_{1} + (1-p) \rho_{2} $, also belongs to $r\text{-}\mathbb{ABSN}$. 

Now, \begin{equation}
     \mathcal{U}  \rho {\mathcal{U}}^{\dagger}  = p  \mathcal{U}  \rho_1 {\mathcal{U}}^{\dagger}+ (1-p) \mathcal{U}  \rho_2 {\mathcal{U}}^{\dagger}
\end{equation}

Since, $\rho_{1} $,  $\rho_{2} \in$ $r\text{-}\mathbb{ABSN}$, therefore $\mathcal{U}  \rho_1 {\mathcal{U}}^{\dagger}$ and $\mathcal{U}  \rho_2 {\mathcal{U}}^{\dagger}$ both belongs to $\mathcal{S}_r$ i.e.,
\begin{equation}
    \text{SN} \Big( \mathcal{U}  \rho_1 {\mathcal{U}}^{\dagger} \Big) \leq r \hspace{0.15cm} \text{and} \hspace{0.15cm} \text{SN} \Big( \mathcal{U}  \rho_2 {\mathcal{U}}^{\dagger} \Big) \leq r
\end{equation}
 Again the set of states $(S_r)$ whose Schmidt number less than or equal to $r$ forms a convex set \cite{terhal2000schmidt},  therefore, their convex combination $ ( p \hspace{0.15cm}  \mathcal{U}  \rho_1 {\mathcal{U}}^{\dagger}  + (1-p) \hspace{0.15cm} \mathcal{U}  \rho_2 {\mathcal{U}}^{\dagger} )$ also belongs to $S_r$ which implies  SN$\Big( p  \mathcal{U}  \rho_1 {\mathcal{U}}^{\dagger}+ (1-p) \mathcal{U}  \rho_2 {\mathcal{U}}^{\dagger} \Big) \leq r$. Therefore, SN$\Big(  \mathcal{U}  \rho {\mathcal{U}}^{\dagger} \Big) \leq r$ for all $\mathcal{U} \in \mathbb{U}$ i.e., $\rho \in$  $r\text{-}\mathbb{ABSN}$. This completes the proof.\qed

\begin{theorem} \label{theorem2}
    The set $r\text{-}\mathbb{ABSN}$ is compact.
\end{theorem} 
\proof First, let us prove that the set $r\text{-}\mathbb{ABSN}$ is closed, i.e., it contains all of its limit points. Let, $\mathcal{S}_{0}$ be an arbitrary limit point of $r\text{-}\mathbb{ABSN}$ (the set $r\text{-}\mathbb{ABSN}$ always has a limit point since we have proved earlier that the set is convex), therefore, if we consider an open ball $B_{\frac{1}{m}}(\mathcal{A}_{0})$ of radius $\frac{1}{m}$ centered on $\mathcal{A}_{0}$, then 
\begin{equation}
   \{ B_{\frac{1}{m}}(\mathcal{A}_{0}) - \mathcal{A}_{0} \} \cap \text{$r\text{-}\mathbb{ABSN}$} \neq \emptyset
\end{equation}
Since, each neighborhood of $\mathcal{A}_{0}$  contains infinitely many points of $r\text{-}\mathbb{ABSN}$, where $\emptyset$ indicates the null set.
Let us now construct a sequence $\{\mathcal{A}_m\}$ of distinct states which belongs to $r\text{-}\mathbb{ABSN}$ such that   $\mathcal{A}_m \rightarrow \mathcal{A}_{0} $ as follows:
\begin{align}
    & \mathcal{A}_1 \in B_1(\mathcal{A}_{0}) \cap \text{$r\text{-}\mathbb{ABSN}$} , \hspace{0.15cm} \mathcal{A}_1 \neq \mathcal{A}_{0}\\ \nonumber
     &  \mathcal{A}_2 \in  B_{\frac{1}{2}}(\mathcal{A}_{0}) \cap \text{$r\text{-}\mathbb{ABSN}$} ,  \hspace{0.15cm} \mathcal{A}_2 \neq \mathcal{A}_{0},  \mathcal{A}_1 \\ \nonumber
      &   .... \in ....  \\ \nonumber
       &   .... \in ....  \\ \nonumber
         &  \mathcal{A}_m \in  B_{\frac{1}{m}}(\mathcal{A}_{0}) \cap \text{$r\text{-}\mathbb{ABSN}$} , \hspace{0.15cm}  \mathcal{A}_m \neq  \mathcal{A}_{0},...,\mathcal{A}_{m-1}
   \end{align}
From the above construction, it is evident that $\mathcal{A}_m \rightarrow \mathcal{A}_{0}$. Therefore for arbitrary global unitary operator $\mathcal{U}$, we have $\mathcal{U}  \mathcal{A}_m {\mathcal{U}}^{\dagger} \rightarrow \mathcal{U}  \mathcal{A}_{0} {\mathcal{U}}^{\dagger}$. Now since $\mathcal{A}_m \in r\text{-}\mathbb{ABSN}$, hence $\mathcal{U}  \mathcal{A}_m {\mathcal{U}}^{\dagger} \in \mathcal{S}_r$. We know that the set of states $(\mathcal{S}_r)$ whose Schmidt number is less than or equal to $r$ forms a closed set \cite{terhal2000schmidt}, therefore it must contains all of its limit points i.e., limit points of $\mathcal{U}  \mathcal{A}_m {\mathcal{U}}^{\dagger}$. Hence, we obtain $\mathcal{U}  \mathcal{A}_{0} {\mathcal{U}}^{\dagger} \in  \mathcal{S}_r$ for all global unitary operators $\mathcal{U}$ i.e., $ \mathcal{A}_{0} \in r\text{-}\mathbb{ABSN}$. Since we have considered $\mathcal{A}_0$ to be an arbitrary limit point of $r\text{-}\mathbb{ABSN}$, it follows that $r\text{-}\mathbb{ABSN}$ contains all of its limit points. Therefore, the set $r\text{-}\mathbb{ABSN}$ is closed.

Again, $r\text{-}\mathbb{ABSN}$ $\subseteq$ $\mathcal{S}_r$. We know that any closed subset of a compact set is compact. Since the set $\mathcal{S}_r$ forms a compact set, the set $r\text{-}\mathbb{ABSN}$ is also compact. \qed

We conclude this subsection by presenting an important property of r-absolute
Schmidt number states. 

 \begin{theorem} \label{theorem4}
  If $\rho \in \mathcal{D}(\mathbb{C}^d \otimes \mathbb{C}^d)$ belongs to $r\text{-}\mathbb{ABSN}$ and $\rho$ majorizes $\sigma \in \mathcal{D}(\mathbb{C}^d \otimes \mathbb{C}^d)$ i.e., $\rho \succ \sigma$, then $\sigma$ also belongs to $r\text{-}\mathbb{ABSN}$.
\end{theorem}
\proof This follows from Birkhoff's theorem \cite{bhatia2013matrix} and \cite{jivulescu2015positive}. 
Let $\rho \in r\text{-}\mathbb{ABSN}$ and suppose that $\rho \succ \sigma$ \cite{nielsen2010quantum}.  
Then, we can express $\sigma$ as  
\begin{equation}
    \sigma = \sum_i p_i\, \mathcal{U}_i \,\rho\, \mathcal{U}_i^\dagger,
\end{equation}
where $\sum_i p_i = 1$ and each $\mathcal{U}_i \in \mathbb{U} (d^2)$ is a global unitary operator.

Since $\rho \in r\text{-}\mathbb{ABSN}$, and the set $r\text{-}\mathbb{ABSN}$ is convex, it follows that  
\[
    \sum_i p_i\, \mathcal{U}_i \,\rho\, \mathcal{U}_i^\dagger \in r\text{-}\mathbb{ABSN}.
\]
Hence, $\sigma \in r\text{-}\mathbb{ABSN}$, which completes the proof. \qed

\subsection{Examples and volume of states with absolute Schmidt number}

Here, we present examples and volume of states that belong to the set $2$-$\mathbb{ABSN}$. 
\begin{example}
Let $\mathbb{PPT}$ denote the set of all PPT states. The set of absolutely PPT states is then defined as  
\begin{equation}
    \mathbb{APPT} = \{\rho \in \mathbb{PPT} \; : \; U \rho U^{\dagger} \in \mathbb{PPT} \ \forall \ U \in \mathbb{U} \}, \label{absolutelyPPTset}
\end{equation}
where $\mathbb{U}$ denotes the set of all unitary operators.  

Consider any state $\rho \in \mathbb{APPT}$ in $\mathcal{D}(\mathbb{C}^3 \otimes \mathbb{C}^3)$. By definition, $U \rho U^{\dagger} \in \mathbb{PPT}$ for all $U \in \mathbb{U} (9)$. It is also known that every PPT state in $\mathcal{D}(\mathbb{C}^3 \otimes \mathbb{C}^3)$ has Schmidt number at most $2$ \cite{sanpera2001schmidt,yang2016all}. Therefore,  
\[
\mathrm{SN}(U \rho U^{\dagger}) \leq 2, \quad \forall \ U \in \mathbb{U}(9).
\]  
Hence, every state in $\mathbb{APPT}$ within $\mathcal{D}(\mathbb{C}^3 \otimes \mathbb{C}^3)$ also belongs to the set $2$-$\mathbb{ABSN}$.
\end{example}
\begin{lemma} \label{lemma1} \cite{mehta1989matrix}
Let $A$ be a Hermitian operator acting on $\mathbb{C}^{d} \otimes \mathbb{C}^{d}$, and define
\[
\alpha := \frac{\operatorname{Tr}(A)}{\sqrt{\operatorname{Tr}(A^{2})}}.
\]
If $\alpha \ge \sqrt{d^{2}-1}$, then $A$ is positive semidefinite.
\end{lemma}
 \proof For a detailed proof, we refer the interested reader to \cite{mehta1989matrix}.

\begin{theorem}
If a bipartite quantum state $\rho \in \mathcal{D}(\mathbb{C}^3 \otimes \mathbb{C}^3)$ satisfies
\[
\operatorname{Tr}(\rho^{2}) \le \frac{1}{8}
\]
then the state $\rho$ necessarily belongs to the $2$-$\mathbb{ABSN}$ set.
\end{theorem}

\begin{proof}
Let $\rho \in \mathcal{D}(\mathbb{C}^{3} \otimes \mathbb{C}^{3})$ and set $A := \rho^{T_B}$ in Lemma~\eqref{lemma1}.  
Since $\operatorname{Tr}(\rho^{T_B}) = 1$ and $\operatorname{Tr}[(\rho^{T_B})^{2}] = \operatorname{Tr}(\rho^{2})$, Lemma~\eqref{lemma1} yields
\[
\frac{\operatorname{Tr}(\rho^{T_B})}{\sqrt{\operatorname{Tr}[(\rho^{T_B})^{2}]}} \ge \sqrt{8}
\quad \Longrightarrow \quad
\rho \ \text{is PPT}.
\]
Equivalently \cite{gurvits2002largest},
\begin{equation} \label{puritycondition}
\operatorname{Tr}(\rho^{2}) \le \frac{1}{8}
\quad \Longrightarrow \quad
\rho \ \text{is PPT}.
\end{equation}

Let $\rho \in \mathcal{D}(\mathbb{C}^{3} \otimes \mathbb{C}^{3})$ be an absolutely PPT state. By definition, the state $U \rho U^{\dagger}$ remains PPT for all global unitary operator $U \in \mathbb{U}(9)$. Since the purity is unitary invariant, we have
\[
\operatorname{Tr}[(U \rho U^{\dagger})^{2}] = \operatorname{Tr}(\rho^{2}) \le \frac{1}{8},
\]
which implies that $U \rho U^{\dagger}$ is PPT for all $U \in \mathbb{U}(9)$.

Again, it is known that every PPT state in $\mathcal{D}(\mathbb{C}^{3} \otimes \mathbb{C}^{3})$ has Schmidt number at most $2$ \cite{sanpera2001schmidt,yang2016all}.  
Hence,
\[
\operatorname{Tr}(\rho^{2}) \le \frac{1}{8}
\quad \Longrightarrow \quad
\mathrm{SN}(U \rho U^{\dagger}) \le 2
\quad \forall\, U \in \mathbb{U}(9),
\]
which shows that $\rho$ belongs to the $2$-$\mathbb{ABSN}$ set.  
This completes the proof.
\end{proof}

\subsection{Construction of witness operators}
Earlier, we proved that the set $r\text{-}\mathbb{ABSN}$ is convex and compact. This property allows the construction of witness operators capable of detecting states that do not belong to $r\text{-}\mathbb{ABSN}$. We define such a witness operator as follows.

\begin{definition}
A Hermitian operator $\widetilde{\mathcal{W}}$ is called a witness for detecting states that do not belong to $r\text{-}\mathbb{ABSN}$ but lie in $\mathcal{S}_r$ if it satisfies
\begin{align}
& \mathrm{Tr}(\widetilde{\mathcal{W}}\rho) \ge 0 \quad \forall \ \rho \in r\text{-}\mathbb{ABSN}, \\
\text{and} \quad 
& \mathrm{Tr}(\widetilde{\mathcal{W}}\sigma) < 0 \quad \text{for at least one } \sigma \in \mathcal{S}_r-r\text{-}\mathbb{ABSN}.
\end{align}
\end{definition}
Therefore, $\widetilde{\mathcal{W}}$ detects those states whose Schmidt number is initially less than or equal to $r$, but increases beyond $r$ under the action of some global unitary operation.\\

Let \( \sigma \in \mathcal{S}_r \) but \( \sigma \notin r\text{-}\mathbb{ABSN} \). Then, there exists a global unitary operator \( \mathcal{U} \) such that the Schmidt number of the transformed state satisfies \( \text{SN}(\mathcal{U} \sigma \mathcal{U}^{\dagger}) > r \). \\

Now, consider a Schmidt number witness \( \mathcal{W} \)~\cite{sanpera2001schmidt} that detects the state \( \mathcal{U} \sigma \mathcal{U}^{\dagger} \), i.e.,
\[
\text{Tr}(\mathcal{W} \, \mathcal{U} \sigma \mathcal{U}^{\dagger}) < 0.
\]

Using the cyclic property of the trace, we obtain
\[
\text{Tr}(\mathcal{U}^{\dagger} \mathcal{W} \mathcal{U} \, \sigma) < 0,
\]
Moreover for all $\rho \in r\text{-}\mathbb{ABSN}$, we have
\begin{equation}
    \text{Tr}(\mathcal{U}^{\dagger} \mathcal{W} \mathcal{U} \, \rho) = \text{Tr}(\mathcal{W} \, \mathcal{U} \rho \mathcal{U}^{\dagger}) \ge 0
\end{equation}
Since $\rho$ belongs to $r\text{-}\mathbb{ABSN}$,  it follows that $\mathcal{U} \rho \mathcal{U}^{\dagger} \in \mathcal{S}_r$. This implies that the operator $\widetilde{\mathcal{W}}= \mathcal{U}^{\dagger} \mathcal{W} \mathcal{U}$ serves as a valid witness for detecting states in $\mathcal{S}_r-r\text{-}\mathbb{ABSN}$.\\
We now present illustrative examples of states that do not belong to $r\text{-}\mathbb{ABSN}$ and can be detected using this above-mentioned witness-based methods.
\begin{example}
    Consider the two qutrit state defined as
\begin{equation}
\rho_1 = p \ket{{\phi}^{+}_1} \bra{{\phi}^{+}_1} + \frac{1 - p}{9} I_9, \label{qutritexample1}
\end{equation}
where $\ket{{\phi}^{+}_1} = \ket{00} \in \mathbb{C}^3 \otimes \mathbb{C}^3$ and $p \in [0,1]$.\end{example}
It is straightforward to verify that $\rho_1$ is separable for $p \in [0,1]$, i.e., its Schmidt number satisfies $\text{SN}(\rho_1) = 1$ throughout the entire range of $p$.\\

Now, consider the action of a global unitary operator defined by
\begin{equation} \label{firstunitary}
\begin{split}
\mathcal{U}_1 =\, & \mathbb{I}_9 - \frac{\sqrt{2} - 1}{\sqrt{2}} 
\left( \ket{00}\bra{00} + \ket{22}\bra{22} \right) \\
& + \frac{1}{\sqrt{2}} \left( \ket{22}\bra{00} - \ket{00}\bra{22} \right).
\end{split}
\end{equation}

Applying this global unitary to $\rho_1$ yields the new state
\begin{equation}
\rho_2 = \mathcal{U}_1 \rho_1 \mathcal{U}_1^\dagger = p \ket{{\phi}^{+}_2} \bra{{\phi}^{+}_2} + \frac{1 - p}{9} I_9,  \label{qutritexample2}
\end{equation}
where $\ket{{\phi}^{+}_2} = \frac{1}{\sqrt{2}}(\ket{00} + \ket{22}) \in \mathbb{C}^3 \otimes \mathbb{C}^3$.\\

Since the set of states with Schmidt number at most $r$ is convex, it follows that $\text{SN}(\rho_2) \le 2$. To determine whether this bound is tight, we employ a Schmidt number witness as proposed in \cite{sanpera2001schmidt}:
\begin{equation}
\mathcal{W}_1 = I_9 - 3 \ket{{\phi}^{+}_3} \bra{{\phi}^{+}_3} ,
\end{equation}
where, $\ket{{\phi}^{+}_3} = \frac{1}{\sqrt{3}}(\ket{00} +\ket{11}+ \ket{22}) \in \mathbb{C}^3 \otimes \mathbb{C}^3$. \\

Evaluating the expectation value of this witness with respect to $\rho_2$ reveals that $\Tr(\mathcal{W}_1 \rho_2) < 0$ whenever $p > \frac{2}{5}$. Hence, we conclude that $\text{SN}(\rho_2)=2$ for $p > \frac{2}{5}$. Consequently, this implies that $\Tr(\mathcal{U}_1^\dagger \mathcal{W}_1 \mathcal{U}_1 \rho_1) < 0$ whenever $p > \frac{2}{5}$, which further indicates that $\rho_1 \notin 1\text{-}\mathbb{ABSN}$.\\

Next, consider the action of a global unitary operator defined by
\begin{equation} \label{secondunitary}
\mathcal{U}_2 =  \begin{pmatrix}
\frac{1}{\sqrt{3}}+ \frac{1}{\sqrt{6}}& 0 & 0 & 0&0 &0&0&0&-\frac{1}{\sqrt{3}}+ \frac{1}{\sqrt{6}} \\
0 & 1 & 0 & 0 & 0& 0&0& 0& 0 \\
0 & 0 & 1 & 0 &0&0&0&0&0 \\
0 & 0 & 0 & 1 & 0& 0& 0& 0&0\\
-\frac{1}{2\sqrt{3}}+\frac{1}{\sqrt{6}} & 0 & 0 & 0 & \frac{1}{\sqrt{2}}& 0& 0&0& \frac{1}{2\sqrt{3}}+\frac{1}{\sqrt{6}}\\
0 & 0 & 0 & 0 & 0& 1& 0& 0&0\\
0 & 0 & 0 & 0 & 0& 0& 1& 0&0\\
0 & 0 & 0 & 0 & 0& 0& 0& 1&0\\
-\frac{1}{2\sqrt{3}}+\frac{1}{\sqrt{6}} & 0 & 0 & 0 & -\frac{1}{\sqrt{2}}& 0& 0&0& \frac{1}{2\sqrt{3}}+\frac{1}{\sqrt{6}}
\end{pmatrix}
\end{equation}

Applying this global unitary to $\rho_2$ gives us the updated state
\begin{equation}
\rho_3 = \mathcal{U}_2 \rho_2\mathcal{U}_2^\dagger = p \ket{{\phi}^{+}_3} \bra{{\phi}^{+}_3} + \frac{1 - p}{9} I_9,
\end{equation}
 Since the set of states with Schmidt number at most $r$ is convex, it follows that $\text{SN}(\rho_3) \le 3$. To determine the exact Schmidt number of $\rho_3$, we employ the following Schmidt number witness \cite{sanpera2001schmidt}:
\begin{equation}
\mathcal{W}_2 = I_9 - \frac{3}{2} \ket{{\phi}^{+}_3} \bra{{\phi}^{+}_3} ,
\end{equation}
Calculating the expectation value of this witness with respect to $\rho_3$ reveals that $\Tr(\mathcal{W}_2\rho_3) < 0$ whenever $p > \frac{5}{8}$. Therefore we can surely conclude that $\text{SN}(\rho_3)=3$ for $p > \frac{5}{8}$. This implies that $\rho_3 \notin 2\text{-}\mathbb{ABSN}$ and $\mathcal{U}_2^\dagger \,\mathcal{W}_2\, \mathcal{U}_2$ certify this.

It is worth mentioning that witness-based methods are experimentally feasible, mainly because a witness can be decomposed into a collection of local observables \cite{guhne2003investigating,Horodecki09,guhne2009entanglement,
bourennane2004experimental,guhne2002detection, ganguly2011entanglement, chruscinski2014entanglement}, and such decompositions typically require only Pauli measurements. Therefore, the witness-based approach provides an experimentally friendly scheme for detecting states that do not belong to $r\text{-}\mathbb{ABSN}$ and hence can be used as a resource.


\subsection{Detection of non-absolute Schmidt number using moments}

The witness-based  approach has  limitations, one of them being that prior knowledge of the state is typically required to construct the witness \cite{guhne2009entanglement,matsunaga2025detecting,guhne2006nonlinear}.  
Another recently developed approach  relies on the moment-based detection schemes that are both experimentally accessible and do not suffer from some of the disadvantages of witness-based
schemes. We first recall the moment-based detection scheme originally developed for bipartite entanglement detection. Building upon this approach, we formulate a moment-based framework for detecting states that belong to $\mathcal{S}_r - r\text{-}\mathbb{ABSN}$.

In Ref.~\cite{elben2020mixed}, Elben \textit{et al.} formulated an entanglement detection technique that relies on the first three moments of the partially transposed density operator. Their result establishes that any bipartite quantum state $\rho_{AB}$ that is positive under partial transposition (PPT) necessarily fulfills the inequality
\begin{equation}
p_2^2 \leq p_3\, p_1 .
\end{equation}
If this inequality is violated, the state cannot be PPT, implying that it is non-positive under partial transposition (NPT) and therefore entangled. This condition is commonly known as the $p_3$-PPT criterion.

For a bipartite state $\rho_{AB}$, the $n$th-order moments \cite{elben2020mixed,yu2021optimal} are defined by
\begin{equation}
p_n := \operatorname{Tr}\!\left[(\rho_{AB}^{T_A})^n\right],
\label{PTmoments}
\end{equation}
where $n \in \mathbb{N}$. In Ref.~\cite{yu2021optimal}, Yu \textit{et al.} extended this approach by incorporating partial transpose moments of higher order, thereby obtaining entanglement detection criteria that are strictly stronger than the original $p_3$-PPT condition. Their framework is based on the construction of Hankel matrices from the sequence of partial transpose moments $\mathbf{p} = (p_1, p_2, \ldots, p_n)$. The resulting matrices, denoted by $H_n(\mathbf{p})$, are symmetric matrices of dimension $(n+1)\times(n+1)$, whose elements are given by
\begin{equation}
\bigl[H_n(\mathbf{p})\bigr]_{ij} := p_{i+j+1}, \quad i,j \in \{0,1,\ldots,n\}.
\label{Hankelmatrices}
\end{equation}

The explicit expressions for the first few orders of Hankel matrices are given by
\begin{equation}
H_1 =
\begin{pmatrix}
p_1 & p_2 \vspace{0.2cm}\\
p_2 & p_3
\end{pmatrix}
\hspace{0.1cm} \text{and} \hspace{0.1cm}
H_2 =
\begin{pmatrix}
p_1 & p_2 & p_3 \vspace{0.2cm}\\
p_2 & p_3 & p_4 \vspace{0.2cm}\\
p_3 & p_4 & p_5
\end{pmatrix},
\label{secondHankelmatrix}
\end{equation}
which illustrate the general structure of these matrices. Motivated by this construction, a generalized necessary criterion for separability can be formulated as
\begin{equation}
\det\!\left[H_n(\mathbf{p})\right] \geq 0 .
\label{Hankelmatrixcondition}
\end{equation}

 The moments of the partially transposed state can be accessed experimentally without resorting to full quantum state tomography. While state tomography is, in principle, independent of any prior knowledge of the quantum state, it generally requires an exponentially large number of measurement settings, rendering it experimentally demanding. By contrast, the present approach relies on estimating simple functionals of the state, which can be efficiently obtained in experiments using shadow tomography techniques \cite{aaronson2018shadow,aaronson2019gentle,huang2020predicting,elben2020mixed}. Furthermore, the number of state copies required in the moment-based method scales only polynomially with the system size, in contrast to tomography-based methods that necessitate an exponential number of copies. This favorable scaling underscores the practical advantage of the moment-based approach, particularly for large quantum systems and near-term experimental platforms \cite{huang2020predicting}.

With this motivation, we now proceed to develop a moment-based framework to detect states that belong to $\mathcal{S}_r - r\text{-}\mathbb{ABSN}$. We introduce a sequence of moments ${s_n}$ and, based on these moments, we develop a moment-based technique for the detection of quantum states belonging to the set $\mathcal{S}_r - r\text{-}\mathbb{ABSN}$.

\begin{definition}
    Consider a linear, $r$-positive but not $r+1$ positive map $\Lambda_{\mathcal{R}}:\mathcal{M}_d \rightarrow \mathcal{M}_d$ and let $\mathcal{U} \in \mathbb{U} (d^2)$ denote a global unitary operator. Then the $n$-th order moments $s_n$ of the $r$ positive map $\Lambda_{\mathcal{R}}$ are formally defined as:
\begin{equation} \label{moments}
s_n := \text{Tr}\, [S_{\mathcal{R}}^n] 
\end{equation}
where, 
\begin{equation}\label{def_S}
    S_{\mathcal{R}} = \frac{(id_A \otimes \Lambda_{\mathcal{R}})(\mathcal{U}\,\rho\,{\mathcal{U}}^{\dagger})}{\text{Tr}[(id_A \otimes \Lambda_{\mathcal{R}})(\mathcal{U} \,\rho \,{\mathcal{U}}^{\dagger})]}
\end{equation}  with $n$ being an integer. 
\end{definition}
With the above definition in place, we now proceed to introduce our proposed criterion for detecting states that belong to $\mathcal{S}_r-r\text{-}\mathbb{ABSN}$.

\begin{theorem} \label{theorem3}
If a bipartite quantum state $\rho_{AB}$ belongs to $r\text{-}\mathbb{ABSN}$, then for all global unitary operators $\mathcal{U} \in \mathbb{U} (d^2)$
\begin{equation}
       \det[H_{m}(s_\mathcal{R})] \ge 0 \label{Hankel}. 
    \end{equation}
    Here, $[H_{m}(s_\mathcal{R})]_{ij} = s_{i+j+1}$ for $i,j \in \{0,1,...,m\}$, $m \in \mathbb{N}$ and $s_i$, $i= 1,2,..,n$  are the i-th moments defined in Eq.~\eqref{moments} corresponding to $r$ positive but not $r+1$ positive map $\Lambda_R$.
\end{theorem}

\proof Let $\rho_{AB} \in r\text{-}\mathbb{ABSN}$, which implies that for every global unitary $\mathcal{U} \in \mathbb{U} (d^2)$, 
the transformed state $\mathcal{U}\rho_{AB}\mathcal{U}^{\dagger}$ lies within the set $\mathcal{S}_r$. Now, consider an $r$-positive, but not $(r+1)$-positive map $\Lambda_{\mathcal{R}}:\mathcal{M}_d \rightarrow \mathcal{M}_d $. 
We introduce the operator
\begin{equation}
    S_{\mathcal{R}}
    =
    \frac{ (\mathrm{id}_A \otimes \Lambda_{\mathcal{R}})(\mathcal{U}\rho_{AB}\mathcal{U}^{\dagger}) }
    { \operatorname{Tr}\!\left[(\mathrm{id}_A \otimes \Lambda_{\mathcal{R}})(\mathcal{U}\rho_{AB}\mathcal{U}^{\dagger})\right] } .
\end{equation}
From the fundamental properties of positive maps~\cite{horodecki1996necessary},  the operator $S_{\mathcal{R}}$ is guaranteed to be positive semidefinite and normalized to unit trace. 
Therefore, if $\{\mathcal{X}_i\}_{i=1}^{d}$ denote its eigenvalues, they necessarily satisfy $\mathcal{X}_i \ge 0$ for all $i=1,2,\dots,d$.

If $s_{\mathcal{R}} = (s_1, s_2, \ldots, s_n)$ denotes the moment vector introduced in Eq.~\eqref{moments}, 
then the associated $(m+1)\times(m+1)$ Hankel matrix is defined entrywise as
\begin{equation}
    [H_m(s_{\mathcal{R}})]_{ij} = s_{i+j+1}, \qquad i,j \in \{0,1,\ldots,m\}.
\end{equation}
Each Hankel matrix $H_m(s_{\mathcal{R}})$ admits the factorization
\begin{equation}
    H_m(s_{\mathcal{R}}) = V_m\, D\, V_m^{T},
\end{equation}
where $V_m$ denotes the corresponding Vandermonde matrix and $D$ is a diagonal matrix containing the associated weights, explicitly given by
\begin{equation}
    V_m = \begin{pmatrix}
1 & 1 & ...&1  \vspace{0.2cm}\\ 
\mathcal{X}_1 & \mathcal{X}_2 & ...& \mathcal{X}_d  \vspace{0.2cm}\\
...&...&...&...&\\
...&...&...&...&\\
\mathcal{X}_1^m & \mathcal{X}_2^m & ...& \mathcal{X}_d^m     
\end{pmatrix} \hspace{0.1cm}\text{and}\hspace{0.1cm} D = \begin{pmatrix}
 \mathcal{X}_1& 0 & ...&0  \vspace{0.2cm}\\ 
0 & \mathcal{X}_2 & ...& 0  \vspace{0.2cm}\\
...&...&...&...&\\
...&...&...&...&\\
0 & 0 & ...&   \mathcal{X}_d  
\end{pmatrix}. \label{secondHankelmatrix}
\end{equation}

For any vector $x = (x_1, \ldots, x_m, x_{m+1}) \in \mathbb{R}^{m+1}$, we have
\begin{equation}
    x\, H_m(s_{\mathcal{R}})\, x^{T}
    = x\, V_m D V_m^{T}\, x^{T}
    = y\, D\, y^{T}
    = \sum_{i=1}^{d} \mathcal{X}_i\, y_i^{2} \ge 0,
\end{equation}
where $y = x V_m = (y_1, y_2, \ldots, y_d)$, with components
\begin{equation}
    y_i = \sum_{j=1}^{m+1} x_j\, \mathcal{X}_i^{\,j-1},
    \qquad i = 1,2,\ldots,d.
\end{equation}
Since $x\, H_m(s_{\mathcal{R}})\, x^{T} \ge 0$ holds for all $x \in \mathbb{R}^{m+1}$, it follows that
\begin{equation}
    H_m(s_{\mathcal{R}}) \ge 0,
\end{equation}
which in particular implies $\det\!\left[ H_m(s_{\mathcal{R}}) \right] \ge 0$. 
This concludes the proof. \qed\\

The above theorem establishes that the condition in Eq.~\eqref{Hankel} provides a necessary condition for a state to be in $r\text{-}\mathbb{ABSN}$. Therefore, any violation of this criterion is sufficient to demonstrate that the state lies in the set $\mathcal{S}_r - r\text{-}\mathbb{ABSN}$.

\subsubsection{Examples:} We now present examples illustrating the utility of the above Theorem \ref{theorem3}.

\begin{example} \label{example3}
Consider the qutrit state defined as
\begin{equation}
\rho_1 = p \ket{{\phi}^{+}_1} \bra{{\phi}^{+}_1} + \frac{1 - p}{9} I_9,
\end{equation}
where $\ket{{\phi}^{+}_1} = \ket{00} \in \mathbb{C}^3 \otimes \mathbb{C}^3$ and $p \in [0,1]$ \end{example}
Consider the action of a global unitary operator defined in Eq.\eqref{firstunitary}.
Applying our criterion from Theorem \ref{theorem3} with the choice $\Lambda_{\mathcal{R}} = \Lambda_{r}$ for $k=1$, where $\Lambda_{r}$ is defined in Eq. \eqref{reduction}, we find that the first order Hankel matrix $\det[H_1(s_{\mathcal{R}})]$ is negative for $p\ge 0.44$. However, upon applying the second-order Hankel Matrix criterion, we find that $\det[H_2(s_{\mathcal{R}})]$ is not positive semidefinite for $p\ge 0.31$. This violation confirms that Theorem \ref{theorem3} successfully indicates that  $\rho_1 \notin  1\text{-}\mathbb{ABSN}$ for $p\ge 0.31$.\\

\begin{example}
    Next, consider the state
\begin{equation}
\rho_2 = p \ket{\phi^{+}_2}\bra{\phi^{+}_2} + \frac{1-p}{9} I_9,
\end{equation}
where $\ket{\phi^{+}_2} = \tfrac{1}{\sqrt{2}}(\ket{00} + \ket{22}) \in \mathbb{C}^3 \otimes \mathbb{C}^3$ and $p \in [0,1]$.\end{example}
We analyze the behaviour of this state under the global unitary transformation specified in Eq.~\eqref{secondunitary}.
Applying the criterion provided by Theorem~\ref{theorem3} with the choice $\Lambda_{\mathcal{R}} = \Lambda_r$ for $k=\tfrac{1}{2}$, where $\Lambda_r$ denotes the reduction map defined in Eq.~\eqref{reduction}, we find that the first-order Hankel determinant $\det[H_1(s_{\mathcal{R}})]$ becomes negative for all $p > \frac{5}{8}$. This violation clearly indicates that Theorem~\ref{theorem3} is capable of certifying that $\rho_2 \notin 2\text{-}\mathbb{ABSN}$ whenever $p > \frac{5}{8}$.


\section{Quantification of non-$r$-absolute Schmidt number}\label{s4}

Having emphasized the properties and detection of states not belonging to $r\text{-}\mathbb{ABSN}$, the usefulness of such states calls for  a 
resource-theoretic  quantification. Here we propose a framework for   the quantification of non-$r$-absolute Schmidt number states. 
Given a state $\rho \notin r\text{-}\mathbb{ABSN}$, the amount of "non-absoluteness" present in it is of fundamental importance for optimal resource utilisation. In this resource theory, the $r$-absolute Schmidt number states are the free states, whereas global unitaries and their convex mixtures are the free operations. Since the non-$r$-absolute Schmidt number states are resourceful, the amount of  "non-absoluteness" dictates the amount of resource. This is usually quantified by a valid \textit{measure}. The use of a measure as a valid quantifier is central to many other resource theories, such as entanglement \cite{plenio2005introduction}, coherence \cite{PhysRevLett.113.140401}, non-Markovianity \cite{PhysRevLett.103.210401,PhysRevLett.105.050403,mallick2024assessing}, non-absolute separability \cite{patra2023resource}, etc. Motivated by these, we develop a measure for non-$r$-absolute Schmidt number states.

A measure of non-$r$-absolute Schmidt number corresponding to a state $\rho$, denoted as $\mathcal{M}(\rho)$  is a resource quantifier with the following properties :\\
\textit{(i) Positivity}: $\mathcal{M}(\rho) =0$ iff $\rho \in r\text{-}\mathbb{ABSN}$, whereas for a state $\sigma \notin r\text{-}\mathbb{ABSN}$, $\mathcal{M}(\sigma) > 0$.\\
\textit{(ii) Invariance under local unitaries}: $\mathcal{M} (U_1 \otimes U_2 \rho U_1^{\dagger}\otimes U_2 ^{\dagger}) = \mathcal{M}(\rho)$.\\
\textit{(iii) Monotonicity under free operations}:$\mathcal{M}(\Lambda(\rho))\le \mathcal{M}(\rho)$, where $\Lambda(\rho)=\sum_{i}p_i U_i \rho U_i^{\dagger}$. \\
\textit{(iv) Convexity}: $\mathcal{M}(\sum_i p_i \rho_i) \le \sum_i p_i \mathcal{M}(\rho_i)$, which implies that the measure $\mathcal{M}(\rho)$ is a convex function of the states.

Any valid measure of non-$r$-absolute Schmidt number should necessarily satisfy the properties (i)-(iii). If, in addition, it satisfies the property (iv), then $\mathcal{M}(\rho)$ can be perceived as a "good" measure. In the following subsections, we proceed to discuss the robustness and witness-based measures of the non-$r$-absolute Schmidt number.
\subsection{Robustness measure of non-$r$-absolute Schmidt number}
Firstly, we introduce the definitions of the $r$-absolute Schmidt robustness, random $r$-absolute Schmidt robustness and the generalized $r$-absolute Schmidt robustness. 
\begin{definition} \label{def2}
    The robustness measure of the non-$r$-absolute Schmidt number for an arbitrary state $\rho \in \mathcal{D}(\mathbb{C}^d \otimes \mathbb{C}^d)$ is defined as the minimum value of $t (\ge 0)$
    such that $\frac{\rho + t \chi}{1+t} \in r \text{-} \mathbb{ABSN}$. Depending on the choice of $\chi$, we have the concepts of $r$-absolute Schmidt robustness, random $r$-absolute Schmidt robustness, and the generalized $r$-absolute Schmidt robustness.
    \begin{itemize}
        \item For $\chi \in r\text{-}\mathbb{ABSN}$, this robustness is called the $r$-absolute Schmidt robustness (denoted by $\mathcal{M}_{AR})$.
        \item For $\chi = \frac{I}{d^2}$, this robustness is called as the random $r$-absolute Schmidt robustness (denoted by $\mathcal{M}_{RR}$).
        \item For an arbitrary state $\chi \in \mathcal{D}(\mathbb{C}^d \otimes \mathbb{C}^d)$, this is called the generalized $r$-absolute Schmidt robustness (denoted by $\mathcal{M}_{GR}$).
    \end{itemize}
    \end{definition}
 Following the definitions, they satisfy $\mathcal{M}_{GR} (\rho) \le \mathcal{M}_{AR} (\rho)\le \mathcal{M}_{RR} (\rho)$. Moreover, since $r\text{-}\mathbb{ABSN} \subset \mathcal{S}_r$, $\mathcal{M}_{GR} \le \mathcal{M}_{\mathcal{S}_r} (\rho)$, where $\mathcal{M}_{\mathcal{S}_r} (\rho)$ is the generalized Schmidt-number robustness formally introduced in \cite{bae2019moreclar}. Unless stated otherwise, we shall use the generalized $r$-absolute Schmidt robustness throughout.

 Now, we explicitly study the properties of the generalized $r$-absolute Schmidt robustness measure through the following theorem.
 \begin{theorem} \label{theoremmeasure1}
     The generalized $r$-absolute Schmidt robustness measure is a "good" measure.
 \end{theorem}
 \begin{proof}
     \textit{(i)} $\mathcal{M}_{GR} (\rho)=0$ iff $\rho \in r$-$\mathbb{ABSN}$:  If $\rho \in r-\mathbb{ABSN}$, then from the definition \ref{def2}, it follows that $t=0$, which implies that $\mathcal{M}_{GR} (\rho) = 0$. The converse statement also holds from the definition \ref{def2}. So, positivity is satisfied.\\
    
\textit{(ii)} $\mathcal{M}_{GR}(\rho)=\mathcal{M}_{GR} (\mathcal{U}\rho \mathcal{U}^{\dagger})$: 
      \begin{itemize}
      \item Let $\mathcal{M}_{GR}(\rho)=t$ and $\chi$ be an arbitrary state $\in \mathcal{D}(\mathbb{C}^d \otimes \mathbb{C}^d)$ such that $\frac{\rho + t \chi}{1+t} \in r \text{-} \mathbb{ABSN}$. Then, by the properties of $r \text{-} \mathbb{ABSN}$, it holds that $\frac{\mathcal{U} \rho \mathcal{U}^{\dagger} + t \mathcal{U} \chi \mathcal{U}^{\dagger}}{1+t} \in r \text{-} \mathbb{ABSN}$. So, the same parameter ($t$) also works for the state $\mathcal{U} \rho \mathcal{U}^{\dagger}$. However, it may not be the optimal one for $\mathcal{U} \rho \mathcal{U}^{\dagger}$. Hence, $\mathcal{M}_{GR} (\mathcal{U}\rho \mathcal{U}^{\dagger}) \le \mathcal{M}_{GR}(\rho)$. 
     \item Now, let $\mathcal{M}_{GR} (\mathcal{U}\rho \mathcal{U}^{\dagger}) = s$, such that $\frac{\mathcal{U} \rho \mathcal{U}^{\dagger} + s \Omega}{1+s} \in r \text{-} \mathbb{ABSN}$ for an arbitrary state $\Omega \in \in \mathcal{D}(\mathbb{C}^d \otimes \mathbb{C}^d)$. Multiplying it by a global unitary $V=\mathcal{U}^{\dagger}$, we obtain $\frac{\rho + s \mathcal{U}^{\dagger }\Omega \mathcal{U}}{1+s} \in r \text{-} \mathbb{ABSN}$. By a similar logic as above, it follows that $\mathcal{M}_{GR}(\rho) \le \mathcal{M}_{GR} (\mathcal{U}\rho \mathcal{U}^{\dagger})$.
     \end{itemize}
     Hence, the equality holds. Invariance under global unitary implies that the measure also remains invariant under local unitaries.\\
     
     \textit{(iv)} $\mathcal{M}_{GR}(\sum_i p_i \rho_i) \le \sum_i p_i \mathcal{M}_{GR}(\rho_i)$:\\
     
     For simplicity, consider two states $\rho_1, \rho_2 \in \mathcal{D}(\mathbb{C}^d \otimes \mathbb{C}^d)$ such that $\mathcal{M}_{GR}(\rho_1)=t_1$ and $\mathcal{M}_{GR}(\rho_2)=t_2$. Let $\chi_1, \chi_2 \in \mathcal{D}(\mathbb{C}^d \otimes \mathbb{C}^d)$ be arbitrary states such that the states $(\rho_1 + t_1 \chi_1)$ and $ (\rho_2 + t_2 \chi_2) \in r \text{-} \mathbb{ABSN}$. Now, from Theorem \ref{theorem1}, $$p (\rho_1 + t_1 \chi_1) + (1-p) (\rho_2 + t_2 \chi_2) \in r \text{-} \mathbb{ABSN},$$  $$ \implies p \rho_1 + (1-p) \rho_2 + (p t_1 \chi_1 + (1-p)t_2 \chi_2) \in r \text{-} \mathbb{ABSN}$$ $\implies p \rho_1 + (1-p) \rho_2 + (p t_1 + (1-p) t_2) \chi \in r \text{-} \mathbb{ABSN}$, \\
     
     where $\chi = \frac{p t_1 \chi_1 + (1-p)t_2 \chi_2}{p t_1 + (1-p) t_2}$.\\
     
     Hence, $\mathcal{M}_{GR}(p \rho_1 + (1-p) \rho_2) \le p t_1 + (1-p)t_2$. \\
     \begin{equation*}
     \begin{split}
         \implies \mathcal{M}_{GR}(p \rho_1 + (1-p) \rho_2) & \le p \mathcal{M}_{GR} (\rho_1) + (1-p) \mathcal{M}_{GR} (\rho_2).
         \end{split}
     \end{equation*}
      Hence, the measure is convex.\\
     

 \textit{(iii)}\quad $\mathcal{M}_{GR}(\Lambda(\rho)) \le \mathcal{M}_{GR}(\rho)$:
\[
\begin{aligned}
\mathcal{M}_{GR}(\Lambda(\rho))
&= \mathcal{M}_{GR}\!\left(\sum_i p_i\, \mathcal{U}_i \rho \mathcal{U}_i^{\dagger}\right) \\
&\overset{(a)}{\le} \sum_i p_i\, \mathcal{M}_{GR}\!\left(\mathcal{U}_i \rho \mathcal{U}_i^{\dagger}\right) \\
&\overset{(b)}{=} \sum_i p_i\, \mathcal{M}_{GR}(\rho) \\
&= \mathcal{M}_{GR}(\rho).
\end{aligned}
\]

where $(a)$ and $(b)$ follow from the properties $(iv)$ and $(ii)$ respectively. Hence, the measure is a monotone under the free operations of global unitaries and their convex mixtures. \\

Thus, $\mathcal{M}_{GR}(\rho)$ is a "good" measure of non-$r$-absolute Schmidt number.
\end{proof}
Analogous to Theorem \ref{theorem4}, we have the following corollary in terms of the generalized $r$-absolute Schmidt robustness measure.
\begin{corollary} \label{cor1}
    If $\rho \in r \text{-}\mathbb{ABSN}$ and $\chi$ be an arbitrary state $ \in \mathcal{D}(\mathbb{C}^d \otimes \mathbb{C}^d)$ such that $\rho \succ \chi$, then $\mathcal{M}_{GR}(\rho) \ge \mathcal{M}_{GR}(\chi)$
    \begin{proof}
        If $\rho \succ \chi$, then there exist a unitary $\mathcal{U}$ such that $\chi = \sum_{i}p_i \mathcal{U}_i \rho \mathcal{U}_i^{\dagger}$. Hence, by the monotonicity property of the generalized $r$-absolute Schmidt robustness, it follows that $\mathcal{M}_{GR}(\rho) \ge \mathcal{M}_{GR}(\chi)$.
    \end{proof}
\end{corollary}

\subsubsection{Application in channel discrimination task}

\begin{remark}
The generalized $r$-absolute Schmidt robustness has an operational interpretation in terms of the quantitative advantage offered by a non-$r$-absolute Schmidt number state over all $r$-absolute Schmidt number states in a suitably chosen channel discrimination task. 
\end{remark}

To illustrate this connection, consider the simplest case of binary channel discrimination, where there is access to one of two quantum channels $\Phi_1, \Phi_2$ acting on a probe system and chosen according to fixed prior probabilities $p_1$ and $p_2$. We compare two scenarios: (I) the probe–ancilla input state is a non-$r$-absolute Schmidt number state, and (II) the input is restricted to the set of 
$r$-absolute Schmidt number states. After the unknown channel acts on the probe subsystem, the task is to optimally distinguish the resulting output states by performing a suitable measurement. The performance of the discrimination task is quantified by the trace distance between the output states. Given a state $\rho$ as input, the probability of correctly guessing the $i$th channel ($i \in 1,2$) is given by \cite{helstrom1969quantum,piani2009all}
\begin{equation}
    p^{\text{guess}}(\{\Phi_i\},\rho)=\max_{\substack{\{M_i\}}} \sum_i \Tr[M_i (id_A \otimes \Phi_i) \rho]
\end{equation}
where $M_i$ denote the $i$th outcome of the measurement. Let $p_{I}^{\text{guess}} (\{p_i,\Phi_i\},\rho)$, where $\rho \notin $r$-\mathbb{ABSN}$ denote the probability achievable in scenario (I), and let $p_{II} ^{\text{guess}} (\{p_i,\Phi_i\})$ denote the optimal probability achievable when the input is constrained to r-absolute Schmidt number states, i.e., 
\begin{equation}
    p_{II} ^{\text{guess}} (\{p_i,\Phi_i\})=\max_{\substack{\sigma \in r-\mathbb{ABSN}}}p_{II} ^{\text{guess}} (\{p_i,\Phi_i\},\sigma)
\end{equation}
Then the ratio $p_{I}^{\text{guess}} (\{p_i,\Phi_i\},\rho)$: $p_{II} ^{\text{guess}} (\{p_i,\Phi_i\})$ is related to the generalized $r$-absolute Schmidt robustness. Specifically, 
\begin{equation}
    \sup_{\substack{\{p_i,\Phi_i\}}}\frac{p_{I}^{\text{guess}} (\{p_i,\Phi_i\},\rho)}{p_{II} ^{\text{guess}} (\{p_i,\Phi_i\})} \le 1+\mathcal{M}_{GR}(\rho)
\end{equation}
where the supremum is over all such channel discrimination tasks. This inequality is strict whenever the input state is genuinely a non-$r$-absolute Schmidt number state. Hence, the generalized $r$-absolute Schmidt robustness provides a quantitative bound on this separation; it characterizes the maximal advantage in distinguishability that a non-$r$-absolute Schmidt number state can offer over all $r$-absolute Schmidt number states. 

This operational interpretation of the generalized $r$-absolute Schmidt robustness can be analogously extended to the general case of multichannel discrimination \cite{bae2019moreclar}.

\subsection{Witness-based measure of non-$r$-absolute Schmidt number}
The convex and compact structure of $r \text{-}\mathbb{ABSN}$ (as exemplified by Theorems \ref{theorem1} and \ref{theorem2} respectively) enables the witness-based quantification of non-$r$-absolute Schmidt number states. Below, we formally define and discuss the properties of this measure.
\begin{definition} \label{def3}
    The witness-based measure of a non-$r$-absolute Schmidt number state $\rho$ is defined as
    \begin{equation}
\begin{split}
       \mathcal{M}_{\mathcal{W}}(\rho)& = \max \left\{0, -\min_{\substack{\mathcal{W}\in \mathbb{M},\, \mathcal{U} \in \mathbb{U}}} \,\Tr(\mathcal{U}^{\dagger} \,\mathcal{W} \,\mathcal{U} \,\rho)\right\}\\
      &=\max \left\{0, -\Tr(\mathcal{U}_{\rho}^{\dagger} \,\mathcal{W}_{\rho} \,\mathcal{U}_{\rho} \,\rho)\right\}\\
      &=\max \left\{0, -\Tr(\mathcal{\widetilde{W}}_{\rho} \,\rho)\right\}\\
      \end{split}
\end{equation}
Here, $\mathbb{M}=\mathbb{W} \cap \mathcal{C}$ is the intersection of the set of $r$-Schmidt number witnesses ($\mathbb{W}$) along with some other set $\mathcal{C}$ such that $\mathbb{M}$ is compact. The compactness is necessary for making $\mathbb{W}$ bounded, so that the minimization is well-defined \cite{brandao2005quantifying}. $\mathcal{U}_{\rho}$ and $\mathcal{W}_{\rho}$ represent the optimal Unitary and the optimal Schmidt number witness, respectively, and $\mathcal{\widetilde{W}}_{\rho}=\mathcal{U}_{\rho}^{\dagger} \mathcal{W}_{\rho} \mathcal{U}_{\rho}$.
\end{definition}
Throughout this section, $O_x$ denotes the optimal operator $O$ for the state $x$, and $\widetilde{O}=\mathcal{U}^{\dagger}O \mathcal{U}$. 
The properties of this measure are discussed in the following theorem.
\begin{theorem} \label{theoremmeasure2}
   The witness-based measure of non-$r$-absolute Schmidt number is a "good" measure. 
\end{theorem}
\begin{proof}
    \textit{(i)} $\mathcal{M}_{\mathcal{W}} (\rho)=0$ iff $\rho \in r$-$\mathbb{ABSN}$: \\
    If $\rho \in r$-$\mathbb{ABSN}, \Tr(\mathcal{U}_{\rho}^{\dagger} \,\mathcal{W}_{\rho} \,\mathcal{U}_{\rho} \,\rho)\ge 0$, and hence it follows from the definition \ref{def3} that $\mathcal{M}_{\mathcal{W}} (\rho)=0$. Conversely, $\mathcal{M}_{\mathcal{W}} (\rho)=0$ indicates that $\Tr(\mathcal{U}_{\rho}^{\dagger} \,\mathcal{W}_{\rho} \,\mathcal{U}_{\rho} \,\rho)\ge 0$ for the optimal unitary and witness, and hence $\rho \in r \text{-} \mathbb{ABSN}$. So, positivity is satisfied.\\
    \textit{(ii)} $\mathcal{M}_{\mathcal{W}} (\rho)=\mathcal{M}_{\mathcal{W}} (\mathcal{U} \rho \mathcal{U}^{\dagger}):$\\
    \begin{equation*}
        \begin{split}
           \mathcal{M}_{\mathcal{W}} (\rho)&=\text{max} \left\{0, -\Tr(\mathcal{\widetilde{W}}_{\rho} \,\rho)\right\}\\ & =\text{max} \left\{0, -\Tr[(\mathcal{U}\mathcal{\widetilde{W}}_{\rho} \mathcal{U}^{\dagger})\,(\mathcal{U}\rho \mathcal{U}^{\dagger})]\right\}\\ & = \text{max} \left\{0, -\Tr[\mathcal{\widetilde{W}}_{\mathcal{U}\rho \mathcal{U}^{\dagger}} \,(\mathcal{U}\rho \mathcal{U}^{\dagger})]\right\}\\ & = \mathcal{M}_{\mathcal{W}} (\mathcal{U}\rho \mathcal{U}^{\dagger})
        \end{split}
    \end{equation*}
    $\mathcal{\widetilde{W}}_{\mathcal{U}\rho \mathcal{U}^{\dagger}}$ is the witness for $\mathcal{U}\rho \mathcal{U}^{\dagger}$. Invariance under global unitary implies that the measure also remains invariant under local unitaries.\\
    
    \textit{(iii)} $\mathcal{M}_{\mathcal{W}}(\Lambda(\rho))\le \mathcal{M}_{\mathcal{W}}(\rho)$:\\
    \begin{equation*}
        \begin{split}
            \mathcal{M}_{\mathcal{W}}(\Lambda(\rho))&=\text{max} \left\{0, -\min_{\substack{\mathcal{W},\, \mathcal{U} }} \,\Tr(\mathcal{U}^{\dagger} \,\mathcal{W} \,\mathcal{U} \,\Lambda(\rho))\right\}\\& = \text{max} \left\{0, -\Tr(\mathcal{U}_{\Lambda(\rho)}^{\dagger} \,\mathcal{W}_{\Lambda(\rho)} \,\mathcal{U}_{\Lambda(\rho)} \,\Lambda(\rho))\right\}
        \end{split}
    \end{equation*}
    where $\mathcal{U}_{\Lambda(\rho)}$ and $\mathcal{W}_{\Lambda(\rho)}$ denote the optimal unitary and witness, respectively, corresponding to the state $\Lambda(\rho)$, for which the trace attains its minimum value. Using $\widetilde{\mathcal{W}}=\mathcal{U}^{\dagger} \mathcal{W} \mathcal{U}$, we have
    \begin{equation*}
        \begin{split}
            -\Tr(\widetilde{\mathcal{W}}_{\Lambda_{\rho}}\Lambda(\rho)) & \le -\min_{\substack{\widetilde{\mathcal{W}}}} \,\Tr(\widetilde{\mathcal{W}} \Lambda(\rho))\\ & =-\sum_{i}p_i \min_{\substack{\widetilde{\mathcal{W}}}} \,\Tr(\widetilde{\mathcal{W}} \mathcal{U}_i \rho \mathcal{U}_i ^{\dagger}) \\ & = -\sum_{i}p_i \Tr(\widetilde{\mathcal{W}_{\rho_i}} \rho_i)
        \end{split}
    \end{equation*}
where $\rho_i = \mathcal{U}_i \rho \mathcal{U}_i^{\dagger}$ and $\widetilde{\mathcal{W}_{\rho_i}}$ is the optimal witness for $\rho_i$.
\begin{equation*}
        \begin{split}
            -\Tr(\widetilde{\mathcal{W}}_{\Lambda_{\rho}}\Lambda(\rho)) & \le - \sum_{i}p_i \Tr[(\mathcal{U}_i \widetilde{\mathcal{W}_{\rho}} \mathcal{U}_i^{\dagger}) (\mathcal{U}_i \rho \mathcal{U}_i^{\dagger})]\\ & = -\sum_{i}p_i \Tr(\widetilde{\mathcal{W}_{\rho}} \rho) \\ & = -\Tr(\widetilde{\mathcal{W}_{\rho}} \rho)
            \end{split}
            \end{equation*}
        By using this inequality in Definition \ref{def3}, we obtain $\mathcal{M}_{\mathcal{W}}(\Lambda(\rho))\le \mathcal{M}_{\mathcal{W}}(\rho)$. Hence, monotonicity is satisfied. \\
        
        \textit{(iv)} $\mathcal{M}_{\mathcal{W}}(\sum_i p_i \rho_i) \le \sum_i p_i \mathcal{M}_{\mathcal{W}}(\rho_i)$:\\
        
        Let $k=\{(\mathcal{U}, \mathcal{W}): \mathcal{W} \in \mathbb{W}, \mathcal{U} \in \mathbb{U}\}$. Defining $\Tr(\mathcal{U}^{\dagger} \mathcal{W} \mathcal{U} \rho)=\Tr(A_k \rho)=a_k(\rho)$, where $A_K=\mathcal{U}^{\dagger}\mathcal{W}\mathcal{U}$. Thus, 
        \( f(\rho) = \min_{\substack{k}} a_k (\rho)\)\,so that $\mathcal{M}_{\mathcal{W}}(\rho)=\max\{0, -f(\rho)\}$. \\
        
        \begin{proposition} \label{prop1}
           $ f(\rho)$ is a concave function, i.e., for two states $\rho_1, \rho_2$, $f(p \rho_1 + (1-p)\rho_2) \ge p f(\rho_1)+(1-p)f(\rho_2)$ where $p$ is the probability.
        \end{proposition}
        \begin{proof}
        \begin{equation*}
            \begin{split}
             f(p \rho_1 + (1-p)\rho_2) &=  \min_{\substack{k}} \Tr[A_k (p \rho_1 + (1-p)\rho_2)]\\ &  = \min_{\substack{k}} p \Tr(A_k \rho_1)+(1-p)\Tr(A_k \rho_2) \\ &
               = \min_{\substack{k}} p a_k(\rho_1)+(1-p)a_k(\rho_2)
            \end{split}
        \end{equation*}
        Since $f(\rho_i)=\min_{\substack{j}}a_j(\rho_i)$, $f(\rho_i)\le a_k(\rho_i)$, for $i \in \{1,2\}$.
        \begin{equation*}
        \begin{split}
        \implies a_k(p \rho_1 + (1-p)\rho_2) & =p a_k(\rho_1)+(1-p)a_k(\rho_2) \\ & \ge p f(\rho_1)+(1-p)f(\rho_2)
        \end{split}
        \end{equation*}
        Since this holds for all $k$, this should hold for the minimum value as well. Hence,
        \begin{equation} \label{fconcave}
        \begin{split}
        & \min_{\substack{k}}a_k(p \rho_1 + (1-p)\rho_2) \ge p f(\rho_1)+(1-p)f(\rho_2)\\ &
        \implies  f(p \rho_1 + (1-p)\rho_2) \ge p f(\rho_1)+(1-p)f(\rho_2)
        \end{split}
        \end{equation}
        Thus, $f(\rho)$ is a concave function.
        \end{proof}
        Next, we move on to prove that if Proposition \ref{prop1} holds, then $\mathcal{M}_{\mathcal{W}}(\rho)$ is a convex function.
        
        Let $g(\rho)= -f(\rho)$. Hence from Eq.\eqref{fconcave} we get
        \begin{equation*}
            g(p \rho_1 + (1-p)\rho_2) \le p g(\rho_1)+(1-p)g(\rho_2)
        \end{equation*}
        Hence $g(\rho)$ is a convex function.
        \begin{equation*}
             \begin{split}
             &\mathcal{M}_{\mathcal{W}}(p\rho_1 + (1-p) \rho_2) \\
             =& \max \left\{0, g(p\rho_1 + (1-p) \rho_2)\right\}\\ 
             \le & \max \left\{0, p g(\rho_1) + (1-p) g(\rho_2)\right\}\\
              \le &  p \max \left\{0, g(\rho_1)\right\} + (1-p) \max \left\{0, g(\rho_2)\right\}\\ 
               = & p \mathcal{M}_{\mathcal{W}} (\rho_1) + (1-p) \mathcal{M}_{\mathcal{W}} (\rho_2)
             \end{split}
        \end{equation*} 
       The second and third lines follow from the convexity of $g(\rho)$ and the maximisation function, respectively. This proves that the witness-based measure $\mathcal{M}_{\mathcal{W}} (\rho)$ is convex.

       Thus, this measure satisfies all the properties \textit{(i)}-\textit{(iv)}. Hence, it is also a "good" measure of non-$r$-absolute Schmidt number.
\end{proof}
Now, we present two supporting examples of a witness-based non-$r$-absolute Schmidt number measure.
\begin{example} \label{example5}
    For pure states $\rho \in \mathcal{D}(\mathbb{C}^d \otimes \mathbb{C}^d)$, the witness-based non-$r$-absolute Schmidt rank measure is $\mathcal{M}_{\mathcal{W}} (\rho)= \frac{d}{r}-1$.\\
    \end{example}
The witness-based measure is given by 
\begin{equation*}
    \mathcal{M}_{\mathcal{W}}(\rho) = \max \left\{0, -\min_{\substack{\mathcal{W}\in \mathbb{M},\, \mathcal{U} \in \mathbb{U}(d^2)}} \,\Tr(\mathcal{U}^{\dagger} \,\mathcal{W} \,\mathcal{U} \,\rho)\right\}
\end{equation*}
Let $\rho = \ket{\psi}\bra{\psi}$ be a pure state. To evaluate this measure, one needs to find the optimal unitary and witness which minimises $\Tr(\mathcal{W}\mathcal{U}\rho \mathcal{U}^{\dagger})$. The optimal unitary should be such that $\mathcal{U}\rho \mathcal{U}^{\dagger}$ should be a maximal Schmidt rank state. Without loss of generality, we take  $\mathcal{U}\rho \mathcal{U}^{\dagger}=\ket{\phi^+}\bra{\phi^+}$, where $\ket{\phi^+}=\frac{1}{\sqrt{d}} \sum_{i} \ket{ii}$. Note that finding the optimal Schmidt number witness for an arbitrary state is a difficult task, in general. However, we use the optimal witness introduced by Sanpera et al. \cite{sanpera2001schmidt} for our purpose. The optimal $r$-Schmidt number witness is given by 
\begin{equation*}
    \mathcal{W}=I_{d^2}-\frac{d}{r} \ket{\phi^+}\bra{\phi^+}.
\end{equation*}
\begin{equation*}
\implies \min_{\substack{\mathcal{W}\in \mathbb{M},\, \mathcal{U} \in \mathbb{U}(d^2)}} \,\Tr(\mathcal{W} \mathcal{U} \,\rho \mathcal{U}^{\dagger}\,) = 1-\frac{d}{r}.
\end{equation*}
    The witness-based measure then becomes $\frac{d}{r}-1$ for all pure states $\in \mathcal{D}(\mathbb{C}^d \otimes \mathbb{C}^d).$ 

\begin{example}
    The witness-based measure of the state given by Eq.\eqref{qutritexample1} is $\mathcal{M}_{\mathcal{W}}(\rho_1)= \max \left\{0, \frac{5p-2}{3}\right\}$. 
\end{example}
As discussed earlier, the state given by Eq.\eqref{qutritexample1} is separable for $p \in [0,1]$.
Calculating the witness-based measure involves finding the optimal witness and the unitary operator which minimises $\Tr(\mathcal{W} \mathcal{U} \rho_1 \mathcal{U}^{\dagger})$ for the state $\rho_1$. For $\rho_1$, this optimal unitary is given by 
\begin{equation*}
\begin{split}
    \mathcal{U} = & \alpha_0 I_9 + \alpha_1 (I_9 \otimes G_3 + G_3 \otimes I_9) + \alpha_2 (I_9 \otimes G_8 + G_8 \otimes I_9) \\ & + \alpha_3 G_3 \otimes G_3 + \alpha_4 (G_3 \otimes G_8 + G_8 \otimes G_3) + \alpha_5 G_8 \otimes G_8 \\ & + \alpha_6 (G_4 \otimes G_5 + G_5 \otimes G_4)
    \end{split}
 \end{equation*}
 where $G_3 = \begin{pmatrix}
1 & 0 & 0\\
0 & -1 & 0\\
0 & 0 & 0
\end{pmatrix}, G_4 = \begin{pmatrix}
0 & 0 & 1\\
0 & 0 & 0\\
1 & 0 & 0
\end{pmatrix} , G_5= \begin{pmatrix}
0 & 0 & -i\\
0 & 0 & 0\\
i & 0 & 0
\end{pmatrix}$ and $G_8 = \frac{1}{\sqrt{3}}\begin{pmatrix}
1 & 0 & 0\\
0 & 1 & 0\\
0 & 0 & -2
\end{pmatrix}$ are the Gell-Mann matrices, with the corresponding coefficients given by $\alpha_0 = \frac{(7 \sqrt{2}+1)}{9 \sqrt{2}}, \alpha_1 = \frac{(1-\sqrt{2})}{6 \sqrt{2}}, \alpha_2 = \frac{(\sqrt{2}-1)}{6 \sqrt{6}}$, $\alpha_3 = \frac{(1-\sqrt{2})}{4 \sqrt{2}}$, $\alpha_4 = \frac{(1-\sqrt{2})}{4 \sqrt{6}}$, $\alpha_5 = \frac{(1-5\sqrt{2})}{12 \sqrt{2}}$, and $\alpha_6 = \frac{-i}{2 \sqrt{2}}$.
The optimal witness is given by $\mathcal{W}=I_9 - 3 \ket{\phi^+}\bra{\phi^+}$. The value of the witness-based measure for these choices of unitary and witness operators then reduces to 
\begin{equation*}
\begin{split}
    \mathcal{M}_{\mathcal{W}}(\rho_1) = &  \max \left\{0, -\min_{\substack{\{\alpha_i\}_{i=0}^6}} \,\Tr(\mathcal{U}^{\dagger} \,\mathcal{W} \,\mathcal{U} \,\rho_1)\right\} \\ & = \max \left\{0, -\frac{(2-5p)}{3}\right\} \\ & = \max \left\{0, \frac{(5p-2)}{3}\right\}
    \end{split}
\end{equation*}
From the above expression, it can be clearly seen that 
      
\[
\mathcal{M}_{\mathcal{W}}\!\left( \rho_{1} \right)
=
\begin{cases}
0, 
& \quad p \le \dfrac{2}{5}, 
\\[8pt]
\dfrac{5p - 2}{3}, 
& \quad p > \dfrac{2}{5}.
\end{cases}
\]

 This shows that $\rho_1 \notin 1-\mathbb{ABSN}$ for $p>\frac{2}{5}$.
\section{$r$-Absolute Schmidt number Channel} \label{s5}
In this section, we introduce the concept of $r$-absolute Schmidt number channels and present criteria for identifying channels that are not $r$-absolutely Schmidt number channels, as such channels can play an important role in quantum communication, distributed quantum computation, entanglement generation, and several other quantum information processing protocols \cite{zhang2025quantum,bae2019moreclar,kues2017chip,cerf2002security}. It is worth noting that the analysis of absolute Schmidt number channels developed here provides a more general characterization of the absolute separating channels introduced in \cite{filippov2017absolutely}. We now proceed with the formal definition of $r$-absolute Schmidt number channels.
\begin{definition}
A quantum channel 
\(\mathcal{S} : \mathcal{D}(\mathcal{H}_{B_1B_2}) \rightarrow \mathcal{D}(\mathcal{H}_{B_1B_2})\) 
is said to be an absolutely \(r\)-Schmidt number channel if, for every bipartite state 
\(\rho_{B_1B_2} \in \mathcal{D}(\mathcal{H}_{B_1B_2})\), the Schmidt number of the output state satisfies
\begin{equation}
\mathrm{SN}\!\left( \mathcal{U} \mathcal{S}( \rho_{B_1B_2} ) \mathcal{U}^{\dagger} \right) \le r.
\end{equation}
for all $\mathcal{U} \in \mathbb{U}$.
\end{definition}
In other words, a quantum channel $\mathcal{S}$ is termed an absolutely \(r\)-Schmidt number channel, if under its action, the output state always belongs to $r$\text{-}$\mathbb{ABSN}$. Throughout the manuscript, we denote the set of all absolutely \(r\)-Schmidt number channels as \( r\text{-}\mathbb{ABSNC} \).

\begin{theorem}
    If $\mathcal{S}_1$ belongs to $r\text{-}\mathbb{ABSNC}$ and $\mathcal{S}_2$ is an unital channel, then $\mathcal{S}_2 \circ \mathcal{S}_1 $ also belongs to $r\text{-}\mathbb{ABSNC}$.
\end{theorem}
\proof Let $\mathcal{S}_1$ belongs to $r\text{-}\mathbb{ABSNC}$, then 
\begin{equation}
    \mathrm{SN}\!\left( \mathcal{U} \mathcal{S}_1( \rho_{B_1B_2} ) \mathcal{U}^{\dagger} \right) \le r
\end{equation}
for all \(\rho_{B_1B_2} \in \mathcal{D}(\mathcal{H}_{B_1B_2})\) and $\mathcal{U} \in \mathbb{U}$. Let $\mathcal{S}_1(\rho_{B_1 B_2}) = \rho^{'}_{B_1 B_2}$, this implies that $\rho^{'}_{B_1 B_2}$ belongs to $r\text{-}\mathbb{ABSN}$. Now consider the action of the concatenated channel:
\begin{equation}
(\mathcal{S}_2 \circ \mathcal{S}_1)(\rho_{B_1B_2})
= \mathcal{S}_2(\rho'_{B_1B_2})
= \rho''_{B_1B_2}.
\end{equation}
Since $\mathcal{S}_2$ is unital, it follows that $\rho'_{B_1B_2}$ majorizes $\rho''_{B_1B_2}$. Therefore, by Theorem~\eqref{theorem4}, we obtain $\rho''_{B_1B_2} \in r\text{-}\mathbb{ABSN}$. Hence, $\mathcal{S}_2 \circ \mathcal{S}_1 \in r\text{-}\mathbb{ABSNC}$, which completes the proof.\\

The importance of this theorem lies in its implication that unital channels cannot drive a quantum state outside the set
$r\text{-}\mathbb{ABSN}$. Consequently, by the contrapositive of this theorem, any channel
that is capable of transforming a state into one lying outside $r\text{-}\mathbb{ABSN}$
must necessarily be non-unital. \qed \\

\subsection{Detection of non-$r$ absolute Schmidt number channel}
In this subsection, we aim to develop a necessary and sufficient condition for detecting when a channel has the absolute Schmidt
number property. Before going to present a necessary and sufficient condition for a channel to belong to \( r\text{-}\mathbb{ABSNC} \), let us discuss a new class of quantum channels termed as $r$-Schmidt number annihilating channel, which has been recently been introduced in \cite{mallick2024characterization}. An \emph{$r$-Schmidt number annihilating channel} is defined as a quantum channel that reduce the Schmidt number of states supported on subsystem $B$, without necessarily reducing the Schmidt number across the bipartition $A|B$. Throughout this manuscript, we consider the case in which subsystem $B$ is itself bipartite, i.e., $B = B_1 B_2$. Based on their structural properties, $r$-Schmidt number annihilating channels can be divided into two classes: \emph{local} and \emph{global}. A channel $\mathcal{S}$ acting on $B$ is called \emph{local} if it can be expressed as a tensor product of channels acting on the individual subsystems, that is,
\[
\mathcal{S} = \mathcal{S}_1 \otimes \mathcal{S}_2,
\]
where $\mathcal{S}_1$ and $\mathcal{S}_2$ act on $B_1$ and $B_2$, respectively. Channels that do not admit such a factorization are termed \emph{global}. Throughout the manuscript, we primarily focus on global $r$-Schmidt number annihilating channels and proceed to give a precise definition.

\begin{definition}
    Let $\mathcal{E}: \mathcal{D}(\mathcal{H}_{B_1B_2}) \rightarrow \mathcal{D}(\mathcal{H}_{B_1B_2})$ be a quantum channel, where dim$(\mathcal{H}_{B_1B_2}) =d$ . We call $\mathcal{E}$ to be a global r-Schmidt number annihilating channel if SN \,$({\mathcal{E}(\rho_{B})}) \le r$ for every $\rho_{B} \in \mathcal{D}(\mathcal{H}_{B_1B_2})$, with, $r <d$.
\end{definition} 

Hence, a global $r$-Schmidt number annihilating channel acts on subsystem $B$ and guarantees that the Schmidt number of any input state within $B$ does not exceed $r$. With the above definition in place, we now present a necessary and sufficient condition for a channel to belong to \( r\text{-}\mathbb{ABSNC} \), restricted to the class of covariant channels.



\begin{theorem} \label{theorem5}
    A covariant channel \(\mathcal{S} : \mathcal{D}(\mathcal{H}_{B_1B_2}) \rightarrow \mathcal{D}(\mathcal{H}_{B_1B_2})\), with dim$(\mathcal{H}_{B_1B_2}) =d$, belongs to $r\text{-}\mathbb{ABSNC}$ if and only if it is global $r$-Schmidt number annihilating channel.
\end{theorem}
\proof  A channel \(\mathcal{S} : \mathcal{D}(\mathcal{H}_{B_1B_2}) \rightarrow \mathcal{D}(\mathcal{H}_{B_1B_2})\) is said to be \emph{covariant} if \cite{wilde2013quantum}  
\begin{equation}\label{cm}
    \mathcal{S}\big(\mathcal{U} \, \rho_{B_1 B_2} \, \mathcal{U}^{\dagger}\big) 
    = \mathcal{U} \, \mathcal{S}(\rho_{B_1 B_2}) \, \mathcal{U}^{\dagger}
\end{equation}
for all \(\mathcal{U} \in \mathbb{U}(d^2)\).  
Since \(\mathcal{S}\) is a global \(r\)-Schmidt number annihilating channel, we have  
\[
\mathrm{SN}\!\left(\mathcal{S}\big(\mathcal{U} \, \rho_{B_1 B_2} \, \mathcal{U}^{\dagger}\big)\right) \le r,
\]
for all \(\mathcal{U} \, \rho_{B_1 B_2} \, \mathcal{U}^{\dagger} \in \mathcal{D}(\mathcal{H}_{B_1B_2})\).  
Using the covariance property \eqref{cm}, we have  
\[
\mathrm{SN}\!\left(\mathcal{U} \, \mathcal{S}(\rho_{B_1 B_2}) \, \mathcal{U}^{\dagger}\right) \le r,
\]
for all \(\rho_{B_1B_2} \in \mathcal{D}(\mathcal{H}_{B_1B_2})\) and \(\mathcal{U} \in \mathbb{U}(d^2)\).  
Hence, \(\mathcal{S} \in r\text{-}\mathbb{ABSNC}\).\\

  Conversely, let $\mathcal{S}$ be a covariant and absolutely \(r\)-Schmidt number channel. Since, $\mathcal{S}$ belongs to  $r\text{-}\mathbb{ABSNC}$, therefore for any pure state $\rho_{B_1B_2}=|\psi\rangle\langle\psi|  \in \mathcal{D}(\mathcal{H}_{B_1B_2})$, we have SN$(\mathcal{U} \mathcal{S}(|\psi\rangle\langle\psi|) \mathcal{U}^{\dagger}) \le r$ for all $\mathcal{U} \in \mathbb{U}(d^2)$. Now by covariance property this implies  SN$(\mathcal{S}(\mathcal{U} |\psi\rangle\langle\psi|\mathcal{U}^{\dagger})) \le r$ for all global unitary $\mathcal{U} \in \mathbb{U}(d^2)$ which further implies that SN$(\mathcal{S}( |\phi\rangle\langle\phi|))\le r$ for all pure state $|\phi\rangle$. Since the set of input states is convex, this relation extends naturally to arbitrary mixed states also and therefore SN$(\mathcal{S}(\rho_{in}))\le r$ for all $\rho_{\text{in}} \in \mathcal{D}(\mathcal{H}_{B_1B_2})$. Hence, $\mathcal{S}$ is a global $r$-Schmidt number annihilating channel. \qed
  \subsection{Example}
Here, we present an illustrative example involving the qutrit depolarizing channel, identifying the parameter regimes in which they fail to be absolutely Schmidt-number channel.
  \begin{example}
      Consider the global qutrit depolarizing channel $ ( {\mathcal{S}_{dep}^{(3)}})$, whose action is given by ,
\begin{equation} \label{depolarzing}
    {\mathcal{S}_{dep}^{(3)}}(\rho) = p {\rho} + \frac{1-p}{9} \text{Tr}(\rho) I_9
\end{equation}
where, $p\in [0,1]$.\end{example} 
The global qutrit depolarizing channel on an arbitrary two-qutrit input state acts as follows:
\begin{equation}
     \omega_{out} = {\mathcal{S}_{dep}^{(3)}}  ({\omega_{B_1 B_2}}) = p {\omega_{B_1 B_2}} + \frac{1-p}{9} \text{Tr}(\omega_{B_1 B_2})I_9 
\end{equation}
 
Now, to conclude that ${\mathcal{S}_{dep}^{(3)}}$ is a global $2$-Schmidt number annihilating channel, we need to verify for all input states $\omega_{B_1B_2}$. Since the set of states whose Schmidt number less than or equal to $r$ forms a convex set, it suffices to check only for pure states \cite{terhal2000schmidt}. Without loss of generality, any two-qutrit pure state can be expressed in its Schmidt decomposition form as follows.

\begin{equation} \label{schmidtdecompositionqutrit}
    \ket{\psi} = \sum_{j=0}^2 \sqrt{q_j} \ket{\phi_j} \otimes \ket{\widetilde{\phi}_j}
\end{equation}
where $\ket{\phi_j}$ and $\ket{\widetilde{\phi}_j}$ form orthonormal bases for the respective subsystems on Bob's side. Consequently, the resulting output state, denoted by $\omega_{out}$ is given by

\begin{equation} \label{FBoutmat}
\scalebox{0.80}{$
\left[
\begin{array}{ccccccccc}
pq_0+\frac{1-p}{9} & 0 & 0 & 0 & p\sqrt{q_0 q_1} & 0 & 0 & 0 & p\sqrt{q_0 q_2} \\[6pt]
0 & \frac{1-p}{9} & 0 & 0 & 0 & 0 & 0 & 0 & 0 \\[6pt]
0 & 0 & \frac{1-p}{9} & 0 & 0 & 0 & 0 & 0 & 0 \\[6pt]
0 & 0 & 0 & \frac{1-p}{9} & 0 & 0 & 0 & 0 & 0 \\[6pt]
p\sqrt{q_0 q_1} & 0 & 0 & 0 & pq_1+\frac{1-p}{9} & 0 & 0 & 0 & p\sqrt{q_1 q_2} \\[6pt]
0 & 0 & 0 & 0 & 0 & \frac{1-p}{9} & 0 & 0 & 0 \\[6pt]
0 & 0 & 0 & 0 & 0 & 0 & \frac{1-p}{9} & 0 & 0 \\[6pt]
0 & 0 & 0 & 0 & 0 & 0 & 0 & \frac{1-p}{9} & 0 \\[6pt]
p\sqrt{q_0 q_2} & 0 & 0 & 0 & p\sqrt{q_1 q_2} & 0 & 0 & 0 & pq_2+\frac{1-p}{9}
\end{array}
\right]
$}
\end{equation}
where $q_{0}, q_{1} \in [0,1]$ and satisfy the normalization condition $q_{0} + q_{1} = 1 - q_{2}$. By applying the criterion of Theorem~\ref{theorem3} to the specific choice ${\Lambda}_{r}$ with $k=\tfrac{1}{2}$, as defined in Eq.~\eqref{reduction}, where ${\Lambda}_{r}$ denotes the map which is $2$-positive but not $(2+1)$-positive, we find that the condition $\det[H_{1}(s_{r})] < 0$ holds whenever $p > \tfrac{5}{8}$ for $q_{0} =q_{1}=q_{2}=\frac{1}{3}$ . Consequently, using the equivalence condition established in Theorem~\ref{theorem5}, we infer that $\mathcal{S}_{\mathrm{dep}}^{(3)}$ fails to be a $2$-absolute-Schmidt-number channel for $p > \tfrac{5}{8}$.

\section{Conclusions}\label{s7} 
Higher-dimensional entanglement, quantified through the Schmidt number, is a crucial resource for various quantum information processing protocols \cite{zhang2025quantum,bae2019moreclar,kues2017chip,cerf2002security}. 
While the Schmidt number remains invariant
under local unitaries, it can, in general, be increased through suitable global unitary operations.
In this work, we introduce the concept of the absolute Schmidt number, which indicates states whose Schmidt number cannot be increased by the action of any global unitary transformation. We then provide a comprehensive characterization of the class of states that possess this invariance property. We present some structural features of states that possess
the absolute Schmidt number.

Since states with a high Schmidt number offer significant operational advantages, a central and practically relevant problem is to detect states that do not exhibit the absolute Schmidt number property. To 
this end, we first develop a witness-based approach for identifying states whose Schmidt number can indeed be enhanced through suitable global unitaries. Our approach is enabled by our proof the convexity and compactness 
of the set of all states exhibiting the absolute Schmidt number property. 
We further illustrate the construction of witness operators  through
suitable examples. 

 While witness-based methods are, in principle, applicable in arbitrary dimensions, their construction in general requires prior knowledge about the target state \cite{matsunaga2025detecting,guhne2009entanglement,guhne2006nonlinear}. On
the other hand, the moment-based framework  \cite{elben2020mixed,yu2021optimal} bypasses the need for full quantum state tomography and can be realizable using shadow tomography techniques \cite{huang2020predicting,aaronson2018shadow,aaronson2019gentle}. 
Here we introduce a moment-based approaches capable of identifying states whose Schmidt number can indeed be enhanced through suitable global unitaries.  We demonstrate our technique through illustrative examples.

Discerning the advantage of states with nonabsolute Schmidt number calls for
a quantitative formulation. Here, we propose a framework for the resource-theoretic quantification of non-$r$-absolute Schmidt number states.
We formulate two measures, {\it viz.},  a witness-based measure and a robustness measure for quantifying the non-$r$-absolute Schmidt number property of a state. We further show that these measures satisfy several desirable properties of valid measures, including positivity, monotonicity under free operations, and convexity. As a significant illustration of the utility of such measures, we demonstrate the quantitative advantage of our robustness measure in the information theoretic task of channel discrimination.

Moving from states to channels, we then introduce and investigate a new class of quantum channels that exhibit an absolute Schmidt number property. Regardless of the input state, the output produced by the channel always belongs to the set $r$\text{-}$\mathbb{ABSN}$. We refer to such channels as \textit{$r$-absolute Schmidt number channels}. After introducing this class, we derive a necessary and sufficient condition for identifying these channels, with our analysis restricted to the class of covariant channels. We provide illustrative examples in support of our detection scheme. 

Our work opens several promising directions for future research. One further objective is to establish necessary and sufficient conditions for identifying states that possess the $r$-absolute Schmidt number property. In addition, determining the largest possible radius ball consisting entirely of $r$-absolute Schmidt number states, for arbitrary $r$, emerges as a natural extension of our present analysis. Moreover, establishing a necessary and sufficient condition for identifying non-absolute Schmidt number channels beyond the class of covariant channels would be an immediate direction for future investigation. Lastly, considering the practical implementability
of our moment-based detection technique of states that do not belong to $r$\text{-}$\mathbb{ABSN}$, its experimental realization emerges
as a natural outcome of our present analysis.



\section{Acknowledgements}
We acknowledge Jofre Abellanet-Vidal, Guillem Müller-Rigat, Albert Rico, and Anna Sanpera for insightful discussions. B.M. acknowledges the DST INSPIRE fellowship program for financial support.
\bibliography{main}

\end{document}